\documentclass[12pt]{article}

\usepackage{fullpage,amsthm,amsmath,makeidx,citesort,graphicx}
\usepackage{amssymb,epsf,latexsym}
\usepackage{amstext,bm,algorithmic}

\newtheorem{theorem}{Theorem}[section]

\newtheorem{prop}[theorem]{Proposition}
\newtheorem{lemma}[theorem]{Lemma}
\newtheorem{definition}[theorem]{Definition}

\newtheorem{algorithm}[theorem]{Algorithm}

\theoremstyle{definition}
\newenvironment{program}{\indent {\bf PROG} }

\newcommand{\mO}{|\mathcal{O}|}
\newcommand{\mA}{|\mathcal{A}|}
\newcommand{\mpA}{\mathbf{p}_{\mA}}
\newcommand{\mqA}{\mathbf{q}_{\mA}}

\newcommand{\post}[2]{
\centering \leavevmode
 \includegraphics[width=#2cm]{#1}
 }

\author{Ji Zhu$^\dagger$ and Mudhakar Srivatsa$^\ddagger$\\
Dept of Electrical and Computer Engg, University of Illinois at Urbana-Champaign$^\dagger$\\
IBM T.J. Watson Research Center$^\ddagger$\\
{\tt \small jizhu1@illinois.edu, msrivats@us.ibm.com}
}

\begin{document}

\title{Quantifying Information Leakage in Finite Order Deterministic Programs \footnote{A shorter version of this paper is submitted to ICC 2011.}}

\maketitle

\begin{abstract}
Information flow analysis is a powerful technique for reasoning about the sensitive information exposed by a program during its execution. While past work has proposed information theoretic metrics (e.g., Shannon entropy, min-entropy, guessing entropy, etc.) to quantify such information leakage, we argue that some of these measures not only result in counter-intuitive measures of leakage, but also are inherently prone to conflicts when comparing two programs $P_1$ and $P_2$ $-$ say Shannon entropy predicts higher leakage for program $P_1$, while guessing entropy predicts higher leakage for program $P_2$. This paper presents the first attempt towards addressing such conflicts and derives solutions for conflict-free comparison of finite order deterministic programs.
\end{abstract}	

\section{Introduction} \label{sec:introduction}
Protecting sensitive and confidential data is becoming more and more important in many fields of human activities, such as electronic commerce, auctions, payments and voting. Information flow analysis is a powerful technique for reasoning about the sensitive information exposed by a program during its execution \cite{QIF_Backes09,wittbold1990information,Clarkson1443196}. Existing approaches to information flow analysis can be broadly classified into two: qualitative and quantitative approach. Qualitative information flow analysis, such as taint tracking \cite{taint1,taint2}, are coarse-grained $-$ often only distinguishing between {\em possible} leakage and {\em no} leakage. 

Recently, quantitative information analysis \cite{QIF_Backes09,Hamadou09,smith2009foundations,pliam2000incomparability} techniques have been proposed to alleviate this problem by offering a more fine-grained quantitative assessment of information leakage. Such techniques adopt information theoretic metrics \cite{2009arXivTeixeira,cachin1997entropy} such as mutual information between the secret/sensitive input to a program and its public output to quantify information leakage, as shown in figure \ref{fig:ProgramLeakage}.
\begin{figure}
\post{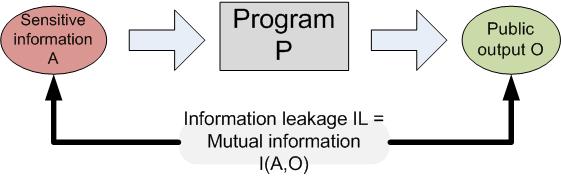}{6}
\label{fig:ProgramLeakage}
\caption{Quantification of Information Leakage in a Program}
\end{figure}
In doing so, several entropy measures have been used to assess mutual information, including, Shannon entropy, Renyi entropy, Guessing entropy (see \cite{Massey94guessingand,smith2009foundations,Hamadou09} for more details), and so on. However, in most past work, the choice of such entropy measure has been ad hoc (mostly driven by sample programs) $-$ sometimes leading to counter-intuitive results. Consider the following two programs (by Smith$^{\cite{smith2009foundations}}$), where the secret input $A$ is uniformly distributed $8k$-bit integer with $k\geq 2$, \& denotes bitwise {\em and} and $0^{7k-1}1^{k+1}$ denotes a binary constant.

\begin{program} {\bf P1}

\label{prog:OisA}

\begin{algorithmic}[H]
\IF{$A\equiv0\mod{8}$}
	\STATE $O=A$
\ELSE \STATE $O=1$
\ENDIF
\end{algorithmic}
\end{program}

\begin{program} {\bf P2} 
\label{prog:OpartofA}

\begin{algorithmic}[H]
\STATE $O = A \ \&\ 0^{7k-1}1^{k+1}$
\end{algorithmic}
\end{program}
Intuitively, one might argue that PROG P1 has much higher information leakage than PROG P2 when $k$ is large, because it reveals complete information about the secret input with probability $\frac{1}{8}$; on the other hand, when $k$ is large, PROG P2 reveals roughly $\frac{1}{8}$ of the number of bits in $A$. However, applying Shannon entropy measure and computing the mutual information $I_1$ between $A$ and $O$ yields a counter intuitive result:
\begin{eqnarray*}
P1: I_1(A,O) &=& -{7\over 8}\log{7\over8} - {1\over 8}\log{1\over 2^{8k}} = k+0.169, \\
P2: I_1(A,O) &=& -2^{k+1}\cdot{2^{7k-1}\over 2^{8k}}\log{2^{7k-1}\over 2^{8k}} = k+1,
\end{eqnarray*}
i.e., leakage by PROG P1 is smaller than leakage by PROG P2, which violates popular consensus in information leakage literature \cite{Hamadou09,smith2009foundations}. Indeed, from a security standpoint, PROG P1 leaves $A$ highly vulnerable to being guessed (e.g., when it is a multiple of 8), while PROG P2 does not (at least for large $k$).

In this paper we argue that past work has failed to address which entropy measure(s) is best suited for quantifying information leakage. Further, this paper shows that some of these entropy based measures (proposed by past work) may be conflicting when they are applied to two programs $P_1$ and $P_2$, i.e., entropy measure $H$ predicts higher leakage for program $P_1$, while entropy measure $H'$ predicts higher leakage for program $P_2$. This paper (to the best of our knowledge) presents the first attempt to analyze different information leakage metrics, show the existence of conflicts in measures proposed by past work and propose a new method for comparing information leakage in finite order deterministic programs. 

{\bf Outline.} The paper is structured as follows. In Section \ref{sec:modelframe}, we present a program model for finite order deterministic programs. Section \ref{sec:conflict} shows the existence of conflicts between leakage measures proposed by past work, followed by our conflict-free leakage metric in Section \ref{sec:solution}. We analyze a few sample programs using our leakage measure in Section \ref{sec:experiment} and conclude in Section \ref{sec:conclusion}.


\section{Model Framework}\label{sec:modelframe}


In this section, we present a formal model for a single-input single-output (SISO) deterministic program and Renyi-entropy based definition of information leakage. A SISO deterministic program is modeled as a group of onto mappings: $O= F_{\mA}(A),\forall \mA\in N^+$, where $A$ is the high (secret/sensitive) input and $O$ is the program output, where $\mA$ denotes the size of the high input set. In other words, for every $\mA\in N^+$, $F_{\mA}$ is an onto mapping from $A\in \mathcal{A}=\{0,1,...\mA-1\}$ to $O\in\mathcal{O}$. We note that $\mA$ acts a tune able security parameter for the program; assuming $\mO$ is fixed, one may be able to increase $\mA$ with the goal of improving the security level of the program. More formally, a SISO deterministic program is defined as follows:

\begin{definition}
A {\em SISO Deterministic Program} is denoted as a 4-tuple $(\mqA, \mA, F_{\mA},\mpA)$, where $\forall \mA\in N^+$, $A$ is a random variable in $\mathcal{A}=\{0,1...\mA-1\}$  with distribution vector $\mqA$, $O=F_{\mA}(A)$ is an onto mapping from $\mathcal{A}$ to $\mathcal{O}$, and $\mpA$ denotes the distribution vector of output $O$ under mapping $F_{\mA}(\cdot)$.

A SISO deterministic program $(\mqA,\mA, F_{\mA},\mpA)$ is said to be a {\em Finite Order SISO Deterministic Program} (FOP) if and only if
\begin{eqnarray*}
\sup_{\mA\in N^+} ||\mathbf{p}_{\mA}||_0 <\infty
\end{eqnarray*}
It is called an {\em Infinite Order SISO Deterministic Program} (IOP) if and only if
\begin{eqnarray*}
\sup_{\mA\in N^+} ||\mathbf{p}_{\mA}||_0 =\infty
\end{eqnarray*}
where $||\mathbf{p}_{\mA}||_0$ is the zero norm of $\mathbf{p}_{\mA}$.
\end{definition}

Unless explicitly specified, in the following portions of this paper, we assume that the secret input $A$ has an uniform prior distribution in $\mathcal{A}$ for any $\mA$.


A key difference between FOPs and IOPs is that the entropy of output $O$ is bounded for FOPs, and so is information leakage. Assuming that $\mO$ is fixed (independent of $\mA$), intuitively the security level of a real FOP will be non-decreasing in $\mA$. In the following portions of this paper we focus on information leakage metrics for FOPs.

Having formalized the program model, we define leakage using Renyi entropy \cite{rrnyi-measures}, which covers most of the entropy metrics adopted by past work on information flow analysis \cite{Massey94guessingand,smith2009foundations,Hamadou09,cover2006elements}, such as Shannon entropy, min-entropy, vulnerability one-guess entropy (proposed by Hamadou et. al, \cite{Hamadou09}), etc. Renyi entropy is defined as follows: 
For a random variable $X$ with distribution $\mathbf{p}=(p_0,p_1,...,p_n)$, its Renyi entropy is defined as:
\begin{eqnarray*}
H_\alpha(X) = \frac{1}{1-\alpha}\log \sum_{i=0}^{n} p_i^\alpha
\end{eqnarray*}
where $\alpha$ is a parameter. In this paper we also apply $H_\alpha(\mathbf{p})$ to denote $H_\alpha(X)$. When $\alpha = 1$, Renyi entropy becomes Shannon entropy; when $\alpha \rightarrow \infty$, $H_\infty(X) = -\log \sup_{i}p_i$ is the min-entropy; when $\alpha$ = 0, $H_0(X)$ denotes the vulnerability one-guess entropy. 

According to general consensus in information flow analysis literate, information leakage ($IL$) of a program $C=(\mqA,\mA,F_{\mA},\mpA)$ (at a given $\mA$) under $\alpha$-Renyi entropy metric is defined as the mutual information $I_\alpha$ between $O$ and $A$:
\begin{eqnarray*}
IL_\alpha(C,\mA) = I_\alpha(O,A) = H_\alpha(O) - H_\alpha(O| A) = H_\alpha(O)
\end{eqnarray*}
where $IL_\alpha(C,\mA)$ denotes a class of information leakage metrics (for different values of $\alpha$) of program $C$. Note that since the program is deterministic $H_\alpha(O | A)=0, \forall\alpha$.



It is worth noting that the mutual information $I_\alpha(O,A)$ may also be defined as $I_\alpha(O,A)=H_\alpha(A)-H_\alpha(A|O)$, which differs from $H_\alpha(O)-H_\alpha(O|A)$ when $\alpha\neq 1$. This alternative definition is not considered here because when $A$ is uniformly distributed, $I_\alpha(O,A)=H_\alpha(A)-H_\alpha(A|O)$ reduces to be Shannon mutual information for all $\alpha$, as shown below:
\begin{eqnarray*}
H_\alpha(A) - H_\alpha(A|O) &=& -\log\mA - \sum_{o\in \mathcal{O}} P(O=o)H_\alpha(A|O=o) \\
&=& -\log\mA + \sum_{o\in \mathcal{O}} P(O=o) \log|\{a:F_{\mA}(a) = o\}| \\
&=& -\log\mA + \sum_{o\in \mathcal{O}} P(O=o) \log\left(\mA P(O=o)\right) = H_1(O)
\end{eqnarray*}

In the next section, we show that this definition of information leakage results in conflicts when comparing two programs. In the subsequent sections we develop solutions for conflict-free comparison of two programs. 

\section{Conflicts in Information Leakage metrics} \label{sec:conflict}

In this section we show several examples of conflicts while comparing two program's information leakage. Recall PROG P1 and PROG P2 from Section \ref{sec:introduction}. Consider the Renyi mutual information of these two PROGs when $\alpha=0, 1, \infty$.
\begin{eqnarray*}
\begin{cases}
IL_0(P1,2^{8k}) = 8k-3, ~IL_0(P2,2^{8k}) = k+1 \\
IL_1(P1,2^{8k}) = k+0.169, ~IL_1(P2,2^{8k}) = k+1 \\
IL_\infty(P1,2^{8k}) = 0.134, ~IL_\infty(P2,2^{8k})=k+1
\end{cases}
\end{eqnarray*}
Note that only the comparing $IL_0(P1,2^{8k})$ and $IL_0(P2,2^{8k})$ agrees with our intuition that P1 leaks much more information than P2; however, comparing $IL_1(P1,2^{8k})$ and $IL_1(P2,2^{8k})$ shows that P1 leaks about the same amount of information as P2; comparing $IL_\infty(P1,2^{8k})$ and $IL_\infty(P2,2^{8k})$ shows that P2 leaks much more information than P1. We see that the leakage measures for different $\alpha$ values conflict with each other, and some of them are even counter-intuitive.

Smith \cite{smith2009foundations} and Hamadou et. al. \cite{Hamadou09} argue that $IL_0$ is more important than $IL_1$ in information flow analysis, because in the above example, $IL_0$ coincides with the intuition but $IL_1$ does not. However, it is not difficult to come up with other examples where $IL_1$ coincides with the intuition but $IL_0$ does not. Consider the following two programs, where the high input $A$ is an uniformly distributed $k$-bit integer with $k\geq 2$ and $L$ is a parameter in $\mathcal{A}$.

\begin{program} {\bf P3} Password Checker

\begin{algorithmic}[H]
\IF{$A = L$}
\STATE $O=1$
\ELSE \STATE $O=0$
\ENDIF
\end{algorithmic}
\end{program}

\begin{program} {\bf P4} Binary Search

\begin{algorithmic}[H]
\IF{$A\geq L$}
\STATE $O=1$
\ELSE \STATE $O=0$
\ENDIF
\end{algorithmic}
\end{program}

Consider $L = \mA/2$. The intuition is that PROG P4 leaks much more information than PROG P3, because when $k$ is large, the probability of $A=L$ becomes so low that PROG P3 leaks almost no information. But PROG P4 always leaks $1$ bit of information, irrespective of $\mA$. Now, consider the Renyi mutual information when $\alpha=0, 1, \infty$:
\begin{eqnarray*}
\begin{cases}
IL_0(P3,2^{k}) = 1,~ IL_0(P4,2^{k}) = 1 \\
IL_1(P3,2^{k}) = H_1(\frac{\mA-1}{\mA},\frac{1}{\mA}),~ IL_1(P4,2^{k}) = 1 \\
IL_\infty(P3,2^{k}) = -\log(1-\frac{1}{\mA}), ~IL_\infty(P4,2^{k})=1
\end{cases}
\end{eqnarray*}
We see that the comparing result when $\alpha=0$ fails to coincide with the intuition, while the comparing results when $\alpha=1$ or $\infty$ match the intuition. The conflict between information leakage metrics for different values of $\alpha$ appears again.

The following lemma indicates that the conflict between different metrics is very common.
\begin{lemma} \label{lem:conflict}\label{lem:finiteorderconflict}
$\forall \alpha\geq0,\beta\geq0,\alpha\neq\beta$, there exists two SISO deterministic programs $C_1 = (\mqA, \mA, F_{\mA},\mpA)$ and $C_2 = (\mqA', \mA, F'_{|\mathcal{A}|},\mpA')$ with $\mqA$ and $\mqA'$ both being uniform distributions in $\mathcal{A}$, such that $\exists D\in R^+$, if $\mA>D$,
\begin{eqnarray}
IL_\alpha(C_1,\mA) &>& IL_\alpha(C_2,\mA) \label{eqn:C1>C2}\\
IL_\beta(C_1,\mA)&<&IL_\beta(C_2,\mA)\label{eqn:C1<C2}
\end{eqnarray}
\end{lemma}
\begin{proof}
The key idea to construct the programs stem from the following property of Renyi entropy $H_\alpha(\mathbf{p})$: $H_\alpha(\mathbf{p})$ is a monotone decreasing function of $\alpha$ for any specific $\mathbf{p}$. Moreover, if $\mathbf{p}$ is uniform, $H_\alpha(\mathbf{p})$ is a constant (independent of $\alpha$); if $\mathbf{p}$ contains a peak probability and a large number of small probabilities, $H_\alpha(\mathbf{p})$ will decreasing quickly as $\alpha$ increases (see \cite{rrnyi-measures} for details). 

First, let us suppose $1<\alpha<\beta\leq\infty$. Pick values $p_0\in(0,1),n\in N^+$ such that
\begin{eqnarray}
{1\over 2^{1-1/\beta}}<&p_0&<{1\over 2^{1-1/\alpha}}\label{eqn:q_0}\\
\log\left[p_0^\beta+{(1-p_0)^\beta\over n^{\beta-1}}\right]&>&1-\beta\label{eqn:beta} \\
 \log\left[p_0^\alpha+{(1-p_0)^\alpha\over n^{\alpha-1}}\right]&<&1-\alpha \label{eqn:alpha}
\end{eqnarray}
We note that one can first pick $p_0$ satisfying $\eqref{eqn:q_0}$; then, to satisfy $ \eqref{eqn:beta}$ and $\eqref{eqn:alpha}$ one simply needs to choose a sufficiently large value for $n$.

Specify the mapping function $F_{\mA}$ for $C_1$ so that the distribution of $O$ is $\mpA = (p_0,p_1,...,p_n)$ with $p_0$ chosen as described above and $p_1=...=p_n={1-p_0\over n}$ for any $\mA>n+1$, and specify the mapping function $F'_{\mA}$ for $C_2$ so that the distribution of $O'$ is $\mpA' = (1/2,1/2)$ for any $\mA$. Then, for any $\mA>n+1$, 
\begin{eqnarray*}
IL_\alpha(C_1,\mA) = H_\alpha(\mpA) = {1\over 1-\alpha}\log\left[ p_0^\alpha+{(1-p_0)^\alpha\over n^{\alpha-1}}\right] &>&1 = H_\alpha(\mpA')= IL_\alpha(C_2,\mA), \\
IL_\beta(C_1,\mA) = H_\beta(\mpA) = {1\over 1-\beta}\log\left[ p_0^\beta+{(1-p_0)^\beta\over n^{\beta-1}}\right] &<&1 = H_\beta(\mpA')= IL_\beta(C_2,\mA),
\end{eqnarray*}
equations $\eqref{eqn:C1>C2}$ and $\eqref{eqn:C1<C2}$ are satisfied.

In the case that $1<\beta<\alpha\leq\infty$, switch the mapping function of $F_{\mA}$ and $F'_{\mA}$ above, so that the distribution of $O_{\mA}$ is $\mathbf{q}$ and the distribution of $O'_{\mA}$ is $\mathbf{p}$, then $\eqref{eqn:C1>C2}$ and $\eqref{eqn:C1<C2}$ are still valid.

Second, suppose $0\leq\alpha<\beta< 1$, pick $2\leq m,n\in N^+$ so that 
\begin{eqnarray}
{1\over 1-\alpha}\log\left[ ({1\over 2})^\alpha+{({1\over2})^\alpha n^{1-\alpha}}\right] > \log m > {1\over 1-\beta}\log\left[ ({1\over 2})^\beta+{({1\over2})^\beta n^{1-\beta}}\right] \label{eqn:alphaless1}
\end{eqnarray}
A sufficiently large $n$ for \eqref{eqn:alphaless1} can make it possible to choose a valid $m$.

Specify the mapping function $F_{\mA}$ for $C_1$ so that the distribution of $O$ is $\mpA = (p_0,p_1,...,p_n)$ with $p_0={1\over2}$ and $p_1=...=p_n={1\over 2n}$ for any $\mA>n+1$, and specify the mapping function $F'_{\mA}$ for $C_2$ so that the distribution of $O'$ is $\mpA' = (1/m,1/m,...1/m)$ for any $\mA$. Then, for any $\mA>n+1$, equations $\eqref{eqn:C1>C2}$ and $\eqref{eqn:C1<C2}$ are satisfied. The case of $0\leq\beta<\alpha<1$ can be proved by switching $F_{\mA}$ and $F'_{\mA}$ as done before.

Third, suppose $\alpha$ and $\beta$ belong to $[0,1]$ and $[1,\infty]$ separately. For example, if $\alpha<1\leq\beta$, pick a value $\beta'$ such that $\alpha<\beta'<1$, and construct $C_1$ and $C_2$ by the same method above, with $\beta'$ in place of $\beta$. Then $\eqref{eqn:C1>C2}$ and $\eqref{eqn:C1<C2}$ can be satisfied because $IL_{\beta'}(C_2,\mA)=IL_{\beta}(C_2,\mA)=IL_\alpha(C_2,\mA)$ as $\mpA'$ is uniform. Equations $\eqref{eqn:C1>C2}$ and $\eqref{eqn:C1<C2}$ in other case of $\alpha$ and $\beta$ can be justified in the same way.

\end{proof}


\section{Quantifying Information Leakage in FOPs} \label{sec:solution}




So far we have shown that some measures of information leakage are not only counter-intuitive, but also introduce conflicts when comparing two programs. In this section we develop a new approach to quantify and compare information leakage in programs. We first sketch the key idea behind our approach. Recall that in FOPs, $\mA$ acts as a security parameter for the program $-$ intuitively, increasing $\mA$ increases the security level of the program (since, $\mO$ is finite and constant $-$ independent of $\mA$). Recall the password checker PROG P3 $-$ observe that increasing the length of the password ($A$) by one bit doubles the security level of the program. 

In this paper we propose that two FOPs $C_1$ and $C_2$ should be compared by examining $\lim_{\mA \rightarrow \infty} IL_\alpha(C_1, \mA)$ $\lim_{\mA \rightarrow \infty} IL_\alpha(C_2, \mA)$. In particular, we show that one can obtain conflict free comparison of programs using a relative leakage metric defined by the ratio $\lim_{\mA \rightarrow \infty} \frac{IL_\alpha(C_1, \mA)}{IL_\alpha(C_2, \mA)}$. Evidently, if the relative leakage metric is 0, then program $C_2$ leaks more information than program $C_1$; if the relative leakage metric is $\infty$, then program $C_1$ leaks more information than program $C_2$. Now, if the relative leakage metric of programs $C_1$ and $C_2$ is a constant $c$ ($c$ $\not=$ 0, $\infty$), one may increase the size of the secret input (namely, $\log \mA$) for program $C_1$ by a constant factor relative to the size of the secret input for program $C_2$ to ensure that the programs $C_1$ and $C_2$ have equal security level; hence, in this case we conclude that the programs $C_1$ and $C_2$ are equal with respect to information leakage. In this section, we formalize this intuition and present a conflict-free approach to comparing information leakage in FOPs.

We first show that for any $C = (\mqA, \mA, F_{\mA},\mpA)$, $IL_{\alpha}(C,\mA)$ is closely related to $||\mpA||_\infty$. 

\begin{lemma} \label{lem:levels}
$\forall 2\leq n\in N$, for any probability distribution vector $\mathbf{p} = (p_1,p_2,...p_n)$ with ordered sequence $||\mathbf{p}||_\infty=p_1\geq p_2\geq...\geq p_n$, then,
\begin{eqnarray} 
\forall 1<\alpha\leq\infty, \lim_{p_1\rightarrow1}{H_\alpha(\mathbf{p})\over 1-p_1} = \frac{\alpha}{\alpha-1}\label{eqn:Hinf}\\
\lim_{p_1\rightarrow1}{H_1(\mathbf{p})\over -(1-p_1)\log(1-p_1)}=1 \label{eqn:H1}\\
\forall \alpha\in(0,1), 
\begin{cases}
\liminf_{p_1\rightarrow1}{H_\alpha(\mathbf{p})\over (1-p_1)^\alpha}>0\\ \limsup_{p_1\rightarrow1}{H_\alpha(\mathbf{p})\over (1-p_1)^\alpha}<\infty  \label{eqn:H0}
\end{cases}
\end{eqnarray}
\end{lemma}
\begin{proof}
Consider \eqref{eqn:Hinf}, use substitution $t=1-p_1$. When $\alpha=\infty$, 
\begin{eqnarray*}
\lim_{p_1\rightarrow1}{H_\infty(\mathbf{p})\over 1-p_1} &=& \lim_{t\rightarrow0}{-\log(1-t)\over t} = 1 \\
\end{eqnarray*}
When $1<\alpha<\infty$, note that $\forall \mathbf{p}$ with $n\geq 2$, $${1\over(n-1)^{\alpha-1}}\leq \sum_{i=2}^{n}(p_i/t)^\alpha\leq 1,$$ so we have
\begin{eqnarray*}
\lim_{p_1\rightarrow1}{H_\alpha(\mathbf{p})\over 1-p_1}&=& {1\over 1-\alpha}\lim_{t\rightarrow0}{(1-t)^\alpha-1+\sum_{i=2}^{n}p_i^\alpha \over t}\\
&=&{1\over 1-\alpha}\left(-\alpha+ \lim_{t\rightarrow0}t^{\alpha-1}\sum_{i=2}^{n}(p_i/t)^\alpha\right) \\
&=& {\alpha\over\alpha-1}
\end{eqnarray*}
Thus, equation \eqref{eqn:Hinf} holds. Next consider \eqref{eqn:H1}.

When $\alpha=1$, note that $\forall \mathbf{p}$ with $n\geq2$, $$0\leq\left|\sum_{i=2}^{n} {p_i\over t}\log{p_i\over t} \right|\leq \log(n-1),$$and we have
\begin{eqnarray*}
&&\lim_{p_1\rightarrow1}{H_1(\mathbf{p})\over (1-p_1)\log(1-p_1)} \\
&=& \lim_{t\rightarrow0}{(1-t)\log(1-t)+\sum_{i=2}^{n}p_i\log p_i\over -t\log t} \\
&=& 1-\lim_{t\rightarrow0}{1\over\log t}\sum_{i=2}^{n}\left({p_i\over t}\log{p_i\over t}  \right) =1
\end{eqnarray*}
Thus, equation \eqref{eqn:H1} holds. Next consider \eqref{eqn:H0}.

When $0<\alpha<1$, note that $\forall \mathbf{p}$ with $n\geq2$, $$1\leq \sum_{i=2}^{n}(p_i/t)^\alpha\leq {(n-1)^{1-\alpha}},$$ so we have
\begin{eqnarray*}
&&\limsup_{p_1\rightarrow1}{H_\alpha(\mathbf{p})\over (1-p_1)^\alpha} \\
&\leq& {1\over 1-\alpha}\left({\lim_{t\rightarrow0} -\alpha t^{1-\alpha}+\limsup_{t\rightarrow0}\sum_{i=2}^{n}(p_i/t)^\alpha } \right) \\
&=& {1\over 1-\alpha}\limsup_{t\rightarrow0}\sum_{i=2}^{n}(p_i/t)^\alpha \leq {(n-1)^{1-\alpha}\over 1-\alpha} <\infty.
\end{eqnarray*}
The other part of \eqref{eqn:H0} can be proved in the same way.
\end{proof}

Define a function $T_\alpha(\cdot)$ for random distributions $\mathbf{p}=(p_1,p_2,...p_n)$ with $p_1\geq p_2\geq...\geq p_n$:
\begin{eqnarray}
T_\alpha(\mathbf{p}) =
\begin{cases}
1-p_1, ~if~\alpha>1\\
-[1-p_1]\log[1-p_1], ~if~\alpha=1\label{eqn:T}\\
[1-p_1]^\alpha,~if~0<\alpha<1
\end{cases}
\end{eqnarray}
Lemma \ref{lem:levels} shows that $\forall \alpha>0, \lim_{||\mathbf{p}||_\infty\rightarrow 1} {H_\alpha(\mathbf{p}) \over T_\alpha(\mathbf{p})}$ is finite. Further, if $2\leq n$ is finite and $\frac{1}{n}\leq ||\mathbf{p}||_\infty<1-\epsilon$ for some $\epsilon>0$, both $H_\alpha(\mathbf{p})$ and $T_\alpha(\mathbf{p})$ will be upper bounded by $\log n$  $<$ $\infty$ and will both be strictly larger than zero. This leads us to Proposition \ref{prop:finite}.

\begin{prop}\label{prop:finite}
For any FOP $C = \{\mqA,\mA, F_{\mA},\mpA\}$ with $\mqA$ being uniform in $\mathcal{A}$, we have
\begin{eqnarray*}
\forall \alpha>0,0<\inf_{\mA} {IL_\alpha(C,\mA)\over T_\alpha(\mpA)}\leq\sup_{\mA}{IL_\alpha(C,\mA)\over T_\alpha(\mpA)}<\infty
\end{eqnarray*}
\end{prop}

Proposition \ref{prop:finite} states that as $\alpha$ is varied, the values of $IL_\alpha$ differ among the levels: $1-||\mpA||_\infty$, $\left(1-||\mpA||_\infty\right)\log\left(1-||\mpA||_\infty\right)$, $\left(1-||\mpA||_\infty\right)^\alpha$. Note that these levels are all related to $1-||\mpA||_\infty$. Intuitively, the rate of convergence of $1-||\mpA||_\infty$ to $0$ determines the security level of a program. We formalize this notion in the following proposition:
\begin{prop}\label{prop:compare}
For any FOPs $C_1 = (\mqA, \mA, F_{\mA},\mpA)$ and $C_2 = (\mqA', \mA, F'_{|\mathcal{A}|},\mpA')$,  with $\mqA$ and $\mqA'$ both being uniform in $\mathcal{A}$. Applying notation $$f_\alpha =\limsup_{|\mathcal{A}|\rightarrow\infty}{IL_\alpha(C_1,\mA)\over IL_\alpha(C_2,\mA)}, ~g_\alpha = \liminf_{|\mathcal{A}|\rightarrow\infty}{IL_\alpha(C_1,\mA)\over IL_\alpha(C_2,\mA)},$$  we have:
\begin{enumerate}
\item 
\begin{eqnarray}
\begin{cases} \label{eqn:f_alpha}
\exists \alpha>0,f_\alpha = 0 \Leftrightarrow \forall\beta>0,f_\beta= 0 \\
\displaystyle\exists \alpha>0, f_\alpha= \infty \Leftrightarrow \forall\beta>0,f_\beta = \infty \\
\displaystyle\exists \alpha>0, 0<f_\alpha<\infty\Leftrightarrow \forall\beta>0, 0<f_\beta < \infty
\end{cases}
\end{eqnarray}
\item 
\begin{eqnarray}
\begin{cases}\label{eqn:g_alpha}
\exists \alpha>0,g_\alpha= 0 \Leftrightarrow \forall\beta>0,g_\beta= 0 \\
\exists \alpha>0, g_\alpha= \infty \Leftrightarrow \forall\beta>0, g_\beta = \infty \\
\exists \alpha>0, 0<g_\alpha<\infty\Leftrightarrow \forall\beta>0, 0<g_\beta< \infty
\end{cases}
\end{eqnarray}
\end{enumerate}
\end{prop}
\begin{proof}

It follows directly from Proposition \ref{prop:finite} that for any $\alpha>0$,
\begin{eqnarray*}
f_\alpha = 0\Leftrightarrow \limsup_{\mA\rightarrow\infty}{T_\alpha(\mpA)\over T_\alpha(\mpA')} = 0, f_\alpha=\infty\Leftrightarrow \limsup_{\mA\rightarrow\infty}{T_\alpha(\mpA)\over T_\alpha(\mpA')} = \infty \\
0<f_\alpha<\infty \Leftrightarrow 0<\limsup_{\mA\rightarrow\infty}{T_\alpha(\mpA)\over T_\alpha(\mpA')} <\infty\\
g_\alpha = 0\Leftrightarrow \liminf_{\mA\rightarrow\infty}{T_\alpha(\mpA)\over T_\alpha(\mpA')} = 0, g_\alpha=\infty\Leftrightarrow \liminf_{\mA\rightarrow\infty}{T_\alpha(\mpA)\over T_\alpha(\mpA')} = \infty \\
0<g_\alpha<\infty \Leftrightarrow 0<\liminf_{\mA\rightarrow\infty}{T_\alpha(\mpA)\over T_\alpha(\mpA')} <\infty\\
\end{eqnarray*}

Note that for any intervals $T, S\subset(0,1)$ and any variables $t\in T, s\in S$, $\forall v\in \{0,\infty\}$,
\begin{eqnarray*}
\limsup_{t\in T,s\in S}{1-t\over 1-s} = v \Leftrightarrow \limsup_{t\in T,s\in S}{(1-t)\log(1-t)\over (1-s)\log(1-s)}=v \Leftrightarrow \limsup_{t\in T,s\in S}{(1-t)^\beta\over (1-s)^\beta}=v,\forall \beta\in(0,1), \\
\liminf_{t\in T,s\in S}{1-t\over 1-s} = v \Leftrightarrow \liminf_{t\in T,s\in S}{(1-t)\log(1-t)\over (1-s)\log(1-s)}=v \Leftrightarrow \liminf_{t\in T,s\in S}{(1-t)^\beta\over (1-s)^\beta}=v,\forall \beta\in(0,1),
\end{eqnarray*}
which indicates that,
\begin{eqnarray*}
\exists \alpha>0, \limsup_{\mA\rightarrow\infty}{T_\alpha(\mpA)\over T_\alpha(\mpA')} = v \Leftrightarrow \forall \beta>0,\limsup_{\mA\rightarrow\infty}{T_\alpha(\mpA)\over T_\alpha(\mpA')} = v \\
\exists \alpha>0, 0<\limsup_{\mA\rightarrow\infty}{T_\alpha(\mpA)\over T_\alpha(\mpA')} < \infty \Leftrightarrow \forall \beta>0,0<\limsup_{\mA\rightarrow\infty}{T_\alpha(\mpA)\over T_\alpha(\mpA')} < \infty \\
\exists \alpha>0, \liminf_{\mA\rightarrow\infty}{T_\alpha(\mpA)\over T_\alpha(\mpA')} = v \Leftrightarrow \forall \beta>0,\liminf_{\mA\rightarrow\infty}{T_\alpha(\mpA)\over T_\alpha(\mpA')} = v \\
\exists \alpha>0, 0<\liminf_{\mA\rightarrow\infty}{T_\alpha(\mpA)\over T_\alpha(\mpA')} <\infty \Leftrightarrow \forall \beta>0,0<\liminf_{\mA\rightarrow\infty}{T_\alpha(\mpA)\over T_\alpha(\mpA')} <\infty 
\end{eqnarray*}
So we conclude that \eqref{eqn:f_alpha} and \eqref{eqn:g_alpha} are valid.
\end{proof}

Now, we are ready to present our solution to compare information leakage of two programs: 
\begin{algorithm}\label{algo:ComparingGeneral}
For any FOPs $C_1=(\mqA, \mA, F_{\mA},\mpA)$ and $C_2 = (\mqA', \mA, F'_{|\mathcal{A}|},\mpA')$,  with $\mqA$ and $\mqA'$ both being uniform in $\mathcal{A}$,
\begin{algorithmic}[H]
\STATE BEGIN PROGRAM
\IF {$f_\infty=\infty$ and $g_\infty>0$} \STATE $C_1$ has a higher leakage than $C_2$. 
\ELSIF{$f_\infty<\infty$ and $g_\infty=0$} \STATE$C_2$ has a higher leakage than $C_1$
\ELSIF {$0<g_\infty\leq f_\infty<\infty$}
\STATE $C_1$ and $C_2$ are on the same leakage level.
\ELSE \STATE $C_1$ and $C_2$ are not comparable.
\ENDIF
\STATE END PROGRAM.
\end{algorithmic}
\end{algorithm}
If ${IL_\infty(C_1,\mA)\over IL_\infty(C_2,\mA)}$ converges as $|\mathcal{A}|\rightarrow\infty$, Algorithm \ref{algo:ComparingGeneral} can be rewritten as:
\begin{algorithm} \ \label{algo:ComparingConverge}

\begin{algorithmic}[H] \label{Alg:convergeCompare}
\STATE BEGIN PROGRAM
\IF {$\lim_{|\mathcal{A}|\rightarrow\infty}\frac{IL_\infty(C_1,\mA)}{IL_\infty(C_2,\mA)}=\infty$} \STATE $C_1$ has a higher leakage than $C_2$
\ELSIF{$\lim_{|\mathcal{A}|\rightarrow\infty}\frac{IL_\infty(C_1,\mA)}{IL_\infty(C_2,\mA)}=0$} \STATE $C_2$ has a higher leakage than $C_1$
\ELSE \STATE $C_1$ and $C_2$ are on the same leakage level.
\ENDIF
\STATE END PROGRAM.
\end{algorithmic}
\end{algorithm}

Due to  Algorithm \ref{Alg:convergeCompare}, it is natural to define the leakage level of a FOP $C$ as the rate of convergence of $IL_\infty(C,\mA)$ as $\mA\rightarrow\infty$:

\begin{definition} {\bf Leakage Level} \label{def:leakagelevel}
For any FOPs $C = (\mqA,\mA,F_{\mA},\mpA)$ with $\mqA$ being uniform in $\mA$, if $IL_\infty(C,\mA)$ converges as $\mA\rightarrow\infty$, then the {\em leakage level} of $C$ is defined to be $\Theta^{\cite{cormen2001introduction}}\Big(IL_\infty(C,\mA)\Big) = \Theta\left(1-||\mpA||_\infty\right)$. 
\end{definition}

We claim that algorithm \ref{algo:ComparingGeneral} (and thus algorithm \ref{algo:ComparingConverge}) offers a conflict-free solution to comparing information leakage of two programs. The proof follows directly from Proposition \ref{prop:compare}. 
We note that in algorithm \ref{algo:ComparingGeneral} that there may be cases wherein two programs are incomparable. However, we claim that it may be impossible to offer a more fine-grained comparison of two programs using Renyi-entropy measure as follows. First, we observe that in Algorithm \ref{algo:ComparingGeneral}, information leakage measures for two are distinguishable if and only if the ratio of their min-entropy leakage metric is either $0=1/\infty$ or $\infty$. The following lemma shows that it is impossible to reduce this ratio to some finite $D$ $<$ $\infty$:
\begin{lemma} \label{lem:conflictfinite}
$\forall D>1$, $\exists \alpha,\beta\in(0,\infty],\alpha\neq\beta$, $\exists$ FOPs $C_1 = (\mqA, \mA, F_{\mA},\mpA)$ and $C_2 = (\mqA', \mA, F'_{|\mathcal{A}|},\mpA')$,  with $\mqA$ and $\mqA'$ both being uniform in $\mathcal{A}$, such that,
\begin{eqnarray*}
\lim_{\mA\rightarrow\infty}{IL_\alpha(C_1,\mA)\over IL_\alpha(C_2,\mA)}&>&D ~but\\
\lim_{\mA\rightarrow\infty}{IL_\beta(C_1,\mA)\over IL_\beta(C_2,\mA)}&<&\frac{1}{D}
\end{eqnarray*}
\end{lemma}
\begin{proof}
We first give an intuitive explanation of the proof of Lemma \ref{lem:conflictfinite} here. Recall from Lemma \ref{lem:levels} that it is feasible to make ${H_\alpha(\mathbf{p})\over H_\beta(\mathbf{p})}$ as large as possible for distribution $\mathbf{p}$ with $||\mathbf{p}||_\infty$ close enough to $1$. This allows us to construct a program $C_1$ with $||\mpA||_\infty$ close to $1$, so that we have $IL_\alpha(C_1,\mA)/IL_\beta(C_1,\mA) > D^2$ (when $\mA$ is large) and a program $C_2$ with $IL_\alpha(C_2,\mA) = IL_\beta(C_2,\mA) = \sqrt{IL_\alpha(C_1,\mA)IL_\beta(C_1,\mA)}$ (when $\mA$ is large). Clearly, the constructed programs $C_1$ and $C_2$ satisfies Lemma \ref{lem:conflictfinite}. 

Here we offer a simple example of $C_1$ and $C_2$. Choose $p_0\in (0,1), 2\leq n\in N$ such that,
\begin{eqnarray*}
2^{-1/D}<p_0<1 \\
\log n > D - {\alpha\over 1-\alpha}\log (1-p_0)
\end{eqnarray*}
Specify $C_1$ so that $\mpA = (p_0,{1-p_0\over n},{1-p_0\over n},...{1-p_0\over n})$ for any $\mA>n+1$, and specify $C_2$ so that $\mpA' = (1/2,1/2)$ for any $\mA$. And consider $0<\alpha<1, \beta=\infty$, then
\begin{eqnarray*}
\lim_{\mA\rightarrow\infty}{IL_\alpha(C_1,\mA)\over IL_\alpha(C_2,\mA)} &=&  {1\over 1-\alpha}\log\left[p_0^\alpha + (1-p_0)^\alpha n^{1-\alpha}\right] \geq  {1\over 1-\alpha}\log\left[ (1-p_0)^\alpha n^{1-\alpha}\right]>D \\
\lim_{\mA\rightarrow\infty}{IL_\beta(C_1,\mA)\over IL_\beta(C_2,\mA)} &=& -\log p_0 <1/D
\end{eqnarray*}
\end{proof}

\section{Experimental Results} \label{sec:experiment}
In this section, we report results obtained by applying our technique to compare information leakage of two programs. We begin by reexamining PROG P4 using our Algorithms. Consider four different parameter values of $L$: $L=\mA/c, L = c\log \mA, L =c\sqrt{\mA}, L=c$ where $c>2$ is certain constant. Then,

\begin{small}
\begin{eqnarray*}
\begin{cases}
L=\mA/c, IL_\infty(P4,\mA) = \log[{c\over c-1}] \\
L =c \log \mA, IL_\infty(P4,\mA) = -\log\left[ 1-{\log \mA\over \mA} \right] \approx {\log\mA\over\mA}\\
L=c\sqrt{\mA}, IL_\infty(P4,\mA) = -\log\left[1-{c\sqrt{\mA}\over \mA}\right]\approx {c \over \sqrt{\mA} }\\
L=c, IL_\infty(P4,\mA) = -\log\left[ 1-{c\over \mA} \right] \approx {c\over\mA}
\end{cases}
\end{eqnarray*}
\end{small}

According to definition \ref{def:leakagelevel}, for PROG P4, when $L=\mA/c$, the leakage level is $\Theta(1)$; when $L=c\log\mA$, the leakage level is $\Theta\left(\log\mA/\mA\right)$; when $L=c\sqrt{\mA}$, the leakage level is $\Theta\left(1/\sqrt{\mA}\right)$; when $L=c$ the leakage level is $\Theta\left(1/\mA\right)$. PROG P4 leaks more information as $L$ decreases. The result matches the intuition of program flow leakage. Indeed as $\mA$ $\rightarrow$ $\infty$ then $L$ = $\frac{\mA}{c}$ leaks non-zero information (e.g., when $c$ = 2 the program leaks one bit of information); while for all other values of $c$ considered above the program leaks almost no information.


Let us now consider program P5 (see below): $A$ is the high input and $1<L\in N^+$ is an integer parameter.

\begin{program}{\bf P5}

\begin{algorithmic}[H]
\STATE $O \equiv A \mod{L}$
\end{algorithmic}
\end{program}
For any value of $1<L\in N^+$, $||\mpA||_\infty =  {\lceil \mA/L\rceil\over \mA}\rightarrow {1\over L}$ as $\mA\rightarrow\infty$, so $IL_\infty(P5,\mA)$ is finite for all $\mA$. Thus P5 with any finite $L$ has leakage level $\Theta(1)$, which indicates that P5 is on the same security level as P4 with $L=\mA/c$.

Let us now consider another program P6 (see below): $A$ is an integer with $k$ bits ($\mA=2^k$), and $0\leq L\leq k$ is an integer parameter.

\begin{program} {\bf P6}

\begin{algorithmic}[H]
\IF{$A$ consists of $L$ bits of $1$ and $k-L$ bits of $0$}
\STATE $O=1$
\ELSE \STATE $O=0$
\ENDIF
\end{algorithmic}
\end{program}
Consider different values of $L$: $L=0,1,2,3...$ Then
\begin{eqnarray*}
IL_\infty(P6,2^k) = -\log\left[1-{{k\choose L}\over 2^k}\right] \approx {{k\choose L}\over 2^k}
\end{eqnarray*}
Because ${k\choose L}/{k\choose L+1} \rightarrow 0$ as $k\rightarrow \infty$, the leakage of PROG P6 increases as $L$ increases. Actually, the leakage level of P6 is $\Theta\left(k^L/2^k\right)$. PROG P6 with $L=0$ has the same leakage level as PROG P4 with $L=c$; PROG P6 with $L=1$ has the same leakage level as PROG P4 with $L=c\log{\mA}$.

We have admit with regret that Algorithm \ref{algo:ComparingGeneral} still unable to distinguish all FOPs, take the following program for example, where $A$ is the high input with $k$-bits and $L\in N$ is an integer parameter.

\begin{program} {\bf P7}

\begin{algorithmic}[H]
\IF{$\log\mA=k$ is even }
\STATE $O\equiv A \mod{2}$
\ELSE
\STATE $O = 1_{\{A=L\}}$
\ENDIF
\end{algorithmic}
\end{program}

PROG P7 has leakage level $\Theta\left(1\right)$ when $\log\mA$ is even, but has leakage level $\Theta\left(1/\mA\right)$ when $\log\mA$ is odd. When comparing P7 ($L=1$) with P4 ($L=c\log\mA$), we have $f_\infty=\infty$ but $g_\infty=0$. It is not applicable to determine a constant leakage level of P7 since it switches between high and low leakage constantly.

\section{Summary} \label{sec:conclusion}
In this paper we point out important drawbacks in past approaches to information-theoretic measures for quantifying program leakage. We show using examples that some of the metrics proposed by past work may not only be counter-intuitive but also conflict with each other. We have presented a novel conflict-free approach to compare information leakage in two programs and show that it may be impossible to derive a more fine-grained comparison using Renyi-entropy based leakage measures. Using several examples we show that the proposed approach vastly outperforms past approaches in matching popular consensus on program information leakage. 


\bibliographystyle{ieeetr}
\bibliography{Report_SummerWork}


\end{document}